\documentclass[10pt,conference]{IEEEtran}
\usepackage{times, graphicx,amsmath,amssymb}
\usepackage{cite}
\usepackage[tight,footnotesize]{subfigure}

\usepackage[tight,footnotesize]{subfigure}
\usepackage{bm,amsmath,amssymb,nicefrac}
\usepackage{algorithmic}
\usepackage{algorithm}
\IEEEoverridecommandlockouts
\usepackage[usenames]{color}
\definecolor{plum}  {rgb}{.4,0,.4}

\newcommand{\modified}[1]{#1}

\newcommand{\xs}{\boldsymbol{\alpha}^*}
\newcommand{\xh}{\hat{\boldsymbol{x}}}
\newcommand{\xt}{\tilde{\boldsymbol{x}}}
\newcommand{\xb}{\boldsymbol{x}}
\newcommand{\that}{\hat{\boldsymbol{f}}}
\newcommand{\ttil}{\tilde{\boldsymbol{f}}}
\newcommand{\tb}{\boldsymbol{f}}
\newcommand{\ybs}{{\boldsymbol{\beta^*}}}
\newcommand{\ybh}{{\boldsymbol{\hat{\beta}}}}
\newcommand{\yb}{\boldsymbol{y}}
\newcommand{\post}{\textit{a posteriori }}

\newcommand{\Phii}{\boldsymbol{\Phi}}
\newcommand{\Ex}{\mathbb{E}}
\newcommand{\w}{\mathbb{Z}_+}
\newcommand{\A}{\Phii}
\newcommand{\At}{\tilde{A}}
\newcommand{\RE}{\mbox{KL}}
\newcommand{\la}{\boldsymbol{g}}
\newcommand{\lb}{\boldsymbol{h}}
\newcommand{\Yb}{\boldsymbol{Y}}
\newcommand{\ii}{{\Lambda}}
\newcommand{\io}{\boldsymbol{I}_\ii}

\newcommand\argmin{\operatornamewithlimits{argmin}}
\newtheorem{theorem}{Theorem}[section]
\newtheorem{definition}[theorem]{Definition}
\newtheorem{proposition}{Proposition}[theorem]

\newtheorem{remark}[theorem]{Remark}
\newtheorem{lemma}[theorem]{Lemma}

\title{Performance Bounds for Expander-Based\\
Compressed Sensing in the Presence\\
of Poisson Noise}

\author{\authorblockN{Sina Jafarpour}
\authorblockA{Computer Science \\ Princeton University
   \\ Princeton, NJ 08540, USA}
\and
\authorblockN{Rebecca Willett, Maxim Raginsky}
\authorblockA{Electrical and Computer Engineering\\
Duke University\\
Durham, NC 27708, USA}
 \and
 \authorblockN{Robert Calderbank\thanks{Copyright 2001 SS\&C. Published in the Proceedings of the Asilomar Conference on Signals, Systems, \& Computers, Nov. 1st-4th, 2009, in Pacific Grove, CA.}}
 \authorblockA{Electrical Engineering \\ Princeton University
    \\ Princeton, NJ 08540, USA}
}

\begin{document}

\maketitle
\begin{abstract}
  This paper provides performance bounds for compressed sensing in the presence of Poisson noise using
  expander graphs. The Poisson
  noise model is appropriate for a variety of
  applications, including low-light imaging and digital streaming,
  where the signal-independent and/or bounded noise models used in the
  compressed sensing literature are no longer applicable. In this paper, we develop a novel sensing paradigm based on expander graphs and propose a
  MAP algorithm for recovering sparse or compressible signals from Poisson observations. The geometry
  of the expander graphs and the positivity of the corresponding sensing matrices play a crucial role in establishing the bounds on the
  signal reconstruction error of the proposed
  algorithm. The geometry of the expander graphs makes them provably
  superior to random dense sensing matrices, such as Gaussian or
  partial Fourier ensembles, for the Poisson noise model. We support our
  results with experimental demonstrations.

\end{abstract}
\section{Introduction}

  The goal of \textit{compressive sampling} or \textit{compressed
sensing} (CS) \cite{donoho,CRT} is to replace conventional sampling by
a more efficient data acquisition framework, requiring fewer
measurements whenever the measurement or compression is costly. This
paradigm is particularly enticing in the context of photon-limited
applications (such as low-light imaging) and digital fountain codes,
since photo-multiplier tubes used in photon-limited imaging are large
and expensive, and the number of packets transmitted via a digital
fountain code is directly tied to coding efficiency. In these and
other settings, however, we cannot directly apply standard methods and
analysis from the CS literature, since these are based on assumptions
of bounded, sparse, or Gaussian noise. Therefore, very little is known
about the validity or applicability of compressive sampling to
photon-limited imaging systems and streaming data communication.

The Poisson model is often
used to model images acquired by photon-counting devices, particularly
when the number of photons is small and a Gaussian approximation is
inaccurate \cite{SnyderCCD}.  Another application is data streaming,
in which streams of
data are transmitted
through a channel with Poisson statistics.

The Poisson model, commonly used to describe
photon-limited measurements and discrete-time memoryless Poisson
communication channels, pose significant theoretical and
practical challenges in the context of CS.  One of the key challenges
is the fact that the measurement error variance scales
with the true intensity of each measurement, so that we cannot assume
uniform noise variance across the collection of measurements.  The
approach considered in this paper hinges, like most CS methods, on
reconstructing a signal from compressive measurements by optimizing a
sparsity-regularized data-fitting expression. In contrast to many CS
approaches, however, we measure the fit of an estimate to the data
using the Poisson log likelihood instead of a squared error term.

In previous work \cite{wr,rhmw}, we showed that a Poisson noise model
combined with conventional dense CS sensing matrices (properly scaled)
yielded performance bounds which were somewhat sobering relative to
bounds typically found in the literature. In particular, we found that
if the number of  photons (or packets) available to sense were held
constant, and if the number of measurements, $m$, was above some
critical threshold, then larger $m$ in general led to larger bounds on
the error between the true and the estimated signals. This can
intuitively be understood as resulting from the low signal-to-noise
ratio of each of the $m$ measurements, which decays with $m$ when the
number of photons (packets) is held constant. 

This paper demonstrates that the bounds developed in previous work can
be improved by considering alternatives to dense sensing matrices
formed by making iid draws from a given probability distribution. In
particular, we show that sensing matrices given by scaled adjacency
matrices of expander graphs have important theoretical characteristics
(especially an $\ell_1$ version of the {\em restricted isometry
property}) which are ideally suited to controlling the performance of
Poisson CS.

Expander graphs have been recently proposed as an
alternative to dense random matrices within the compressed sensing
framework, leading to computationally efficient recovery algorithms
\cite{Sina, newIndyk, BIR}. The approach described in this paper consists of the following key
elements:
\begin{itemize}
\item expander sensing matrices and the RIP-1 associated with them;
\item a reconstruction objective function which explicitly
incorporates the Poisson likelihood;
\item a collection of candidate estimators; and
\item a penalty function defined over the collection of candidates
which satisfies the Kraft inequality and which can be used to promote
sparse reconstructions.
\end{itemize}

 \section{Compressed Sensing using Expander Graphs}
 \label{sec:expander}
 We start by defining an \textit{expander
  graph}.
 \begin{definition}[Expander Graph]
\label{genexpand}
A $(k,\epsilon)$-\textit{expander graph}
is a {\textit{bipartite}} graph $ V=(A,B), |A|=$ $n,~ |B|=$ $m$, where
$A$ is the set of variable nodes and $B$ is the set of parity (or check) nodes,
which is unbalanced, i.e $m=o(n)$, and is left regular with left
degree $d$, such that for any $S \subset A$ with $|S|\leq k$ the
set of neighbors ${\cal N}(S)$ of $S$ has size $|{\cal
  N}(S)|>(1-\epsilon)d|S|$ .
\end{definition}

Expander graphs have been recently proposed as a means of constructing efficient
compressed sensing algorithms \cite{Sina,newIndyk,BIR}.  Figure~\ref{exp1} illustrates 
such a graph. The following proposition, proved using probabilistic methods, states that expander graphs are
optimal in terms of the number of measurements required for
compressive sampling:

\begin{figure}[!t]
\centering
\includegraphics[width=3.5in]{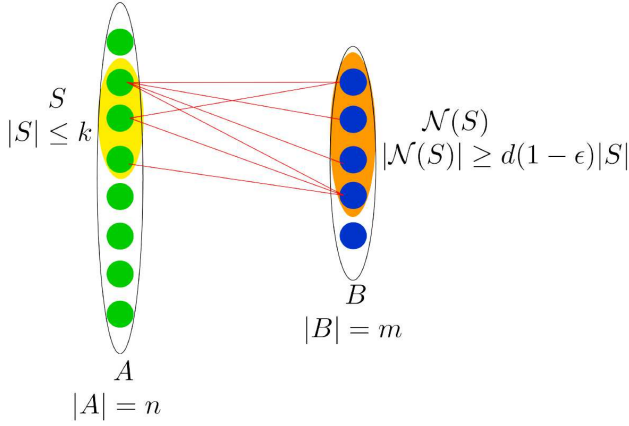}
\caption{A $(k,\epsilon)$-expander graph. In this example, the green
  nodes correspond to $A$, the blue nodes correspond to $B$, the yellow
  oval corresponds to the set $S \subset A$, and the orange oval corresponds to the
  set ${\cal N}(S) \subset B$. There are three colliding
  edges.}
\label{exp1}
\end{figure}

\begin{proposition}
  For any $1 \leq k \leq \frac{n}{2}$ and any positive $\epsilon$,
  there exists a $(k,\epsilon)$-expander graph with left degree
  $d=O\left(\frac{\log(\frac{n}{k})}{\epsilon}\right),$ and right set
  size $m=O\left(\frac{k\log(\frac{n}{k})}{\epsilon^2}\right).$
\end{proposition}

\noindent One reason why expander graphs are good sensing
candidates is that the adjacency matrix of any expander graph almost
preserves the $\ell_1$ norm of any sparse vector (RIP-1). Berinde et al have shown that the RIP-1 property can
be derived from the expansion property \cite{newIndyk}. In Section \ref{sec:result} we
exhibit the role this property plays in the performance of the maximum
\textit{a posteriori} (MAP) estimation algorithm for recovering sparse
vectors in the presence of the Poisson noise.

\begin{proposition}[RIP-1 property of the expander graphs]
\label{RIP}
Let $A$ be the $m \times n$ adjacency matrix of a $(k,\epsilon)$
expander graph $G$. Then for any $k$-sparse vector $x\in \mathbb{R}^n$
we have:
\begin{equation}
\label{ripeq}
(1-2\epsilon)d \|x\|_1  \leq \|Ax\|_1 \leq d~ \|x\|_1
\end{equation}
\end{proposition}
The following theorem is a direct consequence of the RIP-1
property. This theorem states that, for any almost $k$-sparse vector\footnote{By ``almost sparsity" we mean that the vector has at most
  $k$ significant entries.}
$u$, if there exists a vector $v$ whose $\ell_1$ norm is close to that of $u$, and if $v$ approximates $u$ in the measurement domain, then $v$
properly approximates $u$. In Section
\ref{sec:result} we show that the proposed MAP decoding algorithm
outputs a vector satisfying the two conditions above, and hence
approximately recovers the desired signal.
\begin{theorem}
\label{piotr}
Let $A$ be the adjacency matrix of a $(k,\epsilon)$-expander and
$u,v$ be two vectors in $\mathbb{R}^n$, such
that
$$
\|u\|_1\geq
  \|v\|_1-\Delta
$$
  for some positive $\Delta$. Let $S$ be
the set of $k$ largest (in magnitude) coefficients of $u$, and
$\overline{S}$ be the set of remaining coefficients. Then $\|u-v\|_1$
is upper-bounded by
$$ 
 \frac{(1-2\epsilon)}{(1-6\epsilon)}\left(2\|u_{\overline{S}}\|_1+\Delta\right)
 + \frac{2}{d(1-6\epsilon)}\|Au-Av\|_1. 
$$
\end{theorem}
\begin{proof}
  Let $y=u-v$, and $\langle S_1, \cdots, S_t\rangle$ be a decreasing
  partitioning of $\overline{S}$ (with respect to coefficient
  magnitudes), such that all sets but (possibly) $S_t$ have size
  $k$. Note that $S_0 = S$. Let $\At$ be a submatrix of $A$ containing
  rows from ${\cal N}(S)$. Then, following the argument of Berinde et al.~\cite{newIndyk}, we have the following inequality:
\begin{equation}
\label{lemma}
\|Au-Av\|_1+2d\epsilon\|y\|_1 \geq (1-2\epsilon)d \|y_S\|_1.
\end{equation}
Now, using the triangle inequality and Eq.~(\ref{lemma}), we obtain
\begin{align*}\nonumber
\|u\|_1 &\geq \|v\|_1-\Delta\\ \nonumber 
&\geq \|u\|_1-2\|u_{\overline{S}}\|_1\\ \nonumber &
\quad\quad+ \|(u-v)\|_1-2\|{(u-v)}_S\|_1-\Delta
\\ \nonumber
&\geq \|u\|_1-2\|u_{\overline{S}}\|_1-\Delta+\|u-v\|_1\\ \nonumber &\quad\quad-\frac{2\|Au-Av\|_1+4d\epsilon\|u-v\|_1}{(1-2\epsilon)d}.
\end{align*}
Rearranging the inequality completes the proof.
\end{proof}

Finally, note that, since the graph is regular, there exists a minimal
set $\ii$ of variable (left) nodes with size at most $m$, such that
its neighborhood covers all of the check nodes, i.e ${\cal
  N}(\ii)={B}$. Let $\io$ be an index vector such that
$$
(\io)_i=\begin{cases}
	1& \mbox{ if } i\in \ii \\ $0$&\mbox{otherwise}
	\end{cases}
	$$
where $(\io)_i$ denotes the $i$th entry of $\io$. Then $A\io \succeq
I_{m\times 1}.$ The role of $\ii$ is to guarantee that recovery
candidates are non-zero vectors in the measurement domain.  This is
crucial in compressed sensing with Poisson noise, and we will explain
this issue in detail in the next sections.

 \section{Compressed Sensing in the presence of Poisson Noise}
 \label{sec:poisson}
 Recall that a signal is defined to be ``\textit{almost $k$-sparse}''
if it has at most $k$ significant entries, while the remaining
entries have near-zero values. Let $\xs_k$ be the best $k$-term
approximation of $\xs$, and $\Phii_{m\times n}$ be the sensing
matrix. Let $\mathbb{Z}_+=\{0,1,2,\cdots\}$. We assume that each entry
of the measured vector $\yb \in \w^m$ is sensed independently
according to a Poisson model:
$$\yb \sim \mbox{Poisson}\left(\Phii \xs\right).$$
That is, for each index $j$ in $\{1,\cdots,m\}$, the random variable
$Y_j$ is sampled from a Poisson distribution with mean $(\Phii\xs)_j$:\begin{equation}
\label{poisson}
\Pr\left[Y_j|(\Phii \xs)_j\right]=
\begin{cases}
	\frac{(\Phii\xs)_j^{Y_j}}{Y_j!}e^{-(\Phii \xs)_j}& \mbox{if } (\Phii\xs )_j \neq 0
	\\ \boldsymbol{\delta}(Y_j) & \mbox{else}
\end{cases}
\end{equation}
where
$$
\boldsymbol{\delta}(Y_j)=
\begin{cases}
	1& \mbox{if } Y_j = 0
	\\ 0 & \mbox{else}
\end{cases}
$$
Note that 
$$
\lim_{(\Phii\xs)_j\rightarrow 0} \frac{(\Phii\xs)_j^{Y_j}}{Y_j!}e^{-(\Phii \xs)_j}=\delta(Y_j).
$$

We use MAP (maximum \post probability) decoding for recovering a good estimate for $\xs$,
given measurements $\yb$ in the presence of the Poisson noise. Let
$$
\Theta=\left\{f_1,\cdots,f_{|\Theta|}   \right\}
$$
be a set of candidate estimates for $\xs$ such that
$$\forall~f_i\in\Theta~:~ \|f_i\|_1=1~,~f_i\succeq 0.$$
We would like to find the best possible \post estimate, given the
observation vector $\yb$. Moreover, to maintain consistency between the maximum
likelihood and the MAP decoding, we impose the requirement that no
candidate MAP estimator can have a zero coordinate if the
corresponding measurement is non-zero. To guarantee this, let
$\lambda\ll \frac{1}{k \log n}$ be a small parameter. We define
\begin{equation}
\label{gamma}
\Gamma=\left\{x_i=f_i+\lambda \io   \right\}.
\end{equation}
Then since $\lambda$ is strictly positive, we will have  $\Phi x\succ 0$ for any estimate $x$ in
$\Gamma$. This allows us to run the MAP decoding over the set $\Gamma$ and output the
(one-to-one) corresponding estimate from $\Theta$.  We show this
precisely in the next section. This relaxation allows the MAP decoding
to work properly and guarantees recovering an estimate from $\Theta$
with expected $\ell_1$ error close to the error of the best estimate
in $\Theta$.

Let $\text{pen}(x)$ be a nonnegative penalty function based on our prior knowledge
about the estimates in $\Gamma$ (or equivalently let
$\widehat{\text{pen}}(\theta)$ be a penalty function over
$\Theta$). The only constraint that we impose on the penalty function
is the {\em Kraft inequality}
$$\sum_{\xb\in \Gamma} e^{-\text{pen}(\xb)} \leq1.$$ 
For instance, we can impose less penalty on sparser signals or construct a penalty based
on any other prior knowledge about the underlying signal. The log-likelihood of the measurement, according to Eq.~(\ref{poisson}), is 
\begin{eqnarray}
\label{loglikelihood}
\log {\cal L}(\yb|\xb)&=& 
\sum_{j=1}^m{\log \Pr\left[y_j|(\Phii \xb)_j\right]}\\ \nonumber
&\propto& \sum_{j=1}^m
-(\Phii \xb)_j +y_j \log \left(\Phii \xb\right)_j. 
\end{eqnarray}
We will show that the maximum 
\post  estimate
\begin{equation}
\label{min}
\xh\doteq \argmin_{\xb \in \Gamma}\left\{ \sum_{j=1}^m (\Phii \xb)_j -y_j
\log \left(\Phii \xb\right)_j+2\text{pen}(\xb) \right\}
\end{equation}
has error close to the error of the best estimate in $\Gamma$. The decoding in (\ref{min}) is a MAP algorithm over the set of estimates
$\Gamma$, where the likelihood is computed according to the Poisson model (\ref{poisson}) and the penalty function corresponds to a negative
log prior on the candidate estimators in $\Gamma$.

 \section{Performance of MAP Recovery on Almost Sparse Signals   }
 \label{sec:result}
 Let $A$ be the $m \times n$ adjacency matrix of a $(2k,1/16)$-expander
with left degree $d$. Also let $ \Phii= \frac{A}{d}$ be the
sensing matrix. From definition of $\io$ and $\Gamma$, and since the
adjacency matrix of any graph only consists of zeros and ones, for any
estimate $\xb \in \Gamma$ we have $\Phii \xb \succeq
\frac{\lambda}{d}$.  Moreover, from the RIP-1 property of the expander
graphs stated in Lemma~(\ref{RIP}) we know that for any signal $\xb$,
$\|\Phii \xb\|_1 \leq \|\xb\|_1$, and $(1-2\epsilon) \|\xb\|_1 \leq \|\Phii\xb\|_1$ for any $k$-sparse signal $\xb$. Hence
by definition of $\Gamma$
\begin{equation}
\label{mona}
\forall~ \xb\in\Gamma:~~\frac{m\lambda}{d} \leq \|\Phii \xb\|_1 \leq 1+\frac{m\lambda}{d}.
\end{equation}
\begin{lemma}
\label{l1}
Let $\A$ be the normalized expander sensing matrix, $\xs$ be the
original $k$-sparse signal and $\xh$ be the minimizer of the
Equation~(\ref{min}). Then
\begin{eqnarray*}
\lefteqn{\|\A(\xs-\xh)\|_1^2} \nonumber \\ &\leq& 2 \left(2+\frac{m\lambda}{d}\right)  \sum_{i=1}^m \left|(\A\xs)_i^{1/2}-(\A\xh)_i^{1/2} \right|^2
\end{eqnarray*}
\end{lemma}
\begin{proof}Let $\ybs=\A\xs$ and $\ybh=\A\xh$. Then
\begin{align*}
& \|\ybs-\ybh\|_1^2 = \sum_{i=1}^m \left|(\ybs)_i-(\ybh)_i\right|^2\\ \nonumber
&\leq \sum_{i,j=1}^m \left|(\ybs)_i^{1/2}-(\ybh)_i^{1/2} \right|^2. \left| (\ybs)_j^{1/2}+(\ybh)_j^{1/2} \right|^2
\\ \nonumber &\leq
2\sum_{i=1}^m\left| (\ybs)_i^{1/2}-(\ybh)_i^{1/2} \right|^2. \sum_{j=1}^m \left| (\ybs)_j+(\ybh)_j \right|
\\ \nonumber &\leq
2 \left(2+\frac{m\lambda}{d}\right) \sum_{i=1}^m \left|(\ybs)_i^{1/2}-(\ybh)_i^{1/2} \right|^2.
\end{align*}
The first and the second inequalities are by Cauchy--Schwarz, while the third inequality is a consequence of the RIP-1 property of the expander graphs (Lemma~\ref{RIP}) and Eq.~(\ref{mona}).
\end{proof}

\begin{lemma}
\label{int}
Given two Poisson parameter vectors $\la,\lb \in \mathbb{R}^m_+$, the
following equality holds:
\begin{align*}
2 \log \frac{1}{\left(\int{\sqrt{p(\Yb|\la)p(\Yb|\lb)}d\nu(\yb)}
    \right)} = \sum_{j=1}^m
  \left((\la)_j^{1/2}-(\lb)_j^{1/2}\right). 
\end{align*}
\end{lemma}
\begin{proof}
The proof follows from expanding the term $\int{\sqrt{p(\Yb|\la)p(\Yb|\lb)}d\nu(\yb)}$, and is provided in \cite{wr}.
\end{proof}

\begin{lemma}
\label{l2}
Let $\A$ be the expander sensing matrix, $\xs$ be the original almost $k$-sparse signal, and $\xh$ be a minimizer in Eq.~(\ref{min}). Finally let $\yb$ be the compressive measurements of $\xs$ in Poisson model. Then
\begin{eqnarray}
\label{18}
\lefteqn{\Ex_{\Yb|\A\xs}\left[ \sum_{i=1}^m
    \left|(\A\xs)_i^{1/2}-(\A\xh)_i^{1/2} \right|^2 \right]}\\
  \nonumber  &\leq& \min_{\xb \in \Gamma} \left[
    \RE\left(p(\Yb|\A\xs)\parallel
      p(\Yb|\A\xb)\right)+2\text{pen}(\xb)\right]. 
\end{eqnarray}
\end{lemma}
\begin{proof}
The proof exploits techniques from Li and Baron \cite{li}, and Kolaczyk and Nowak \cite{nowak}.
\end{proof}
Now we show that in Poisson setting for all estimates $\xb$ in $\Gamma$, the relative entropy term $\RE\left(p(\Yb|\A\xs)\parallel p(\Yb|\A\xb)\right)$ is upper bounded by the squared $\ell_1$ norm of $\xs - \xb$:
\begin{lemma}
\label{l3}
For any estimate $\xb\in \Gamma$ the following inequality holds:
$$
\RE\left(p(\Yb|\A\xs)\parallel p(\Yb|\A\xb)\right)\leq \frac{d\|\xs-\xb\|_1^2}{\lambda}.
$$
\begin{proof}
\begin{eqnarray}
\nonumber \lefteqn{\RE\left(p(\Yb|\A\xs)\parallel p(\Yb|\A\xb)\right)}
\\ \nonumber &\leq& \sum_{j=1}^m
(\A\xs)_j\left(\frac{(\A\xs)_j}{(\A\xb)_j}-1\right)\\
 \nonumber&&-(\A\xs)_j+(\A\xb)_j
\\ \nonumber &=& \sum_{j=1}^m \frac{1}{(\A\xb)_j}|(\A\xs-\A\xb)_j|^2
\\ \nonumber &\leq& \frac{d}{\lambda}\|\A(\xs-\xb)\|_2^2\\ \nonumber
&\leq& \frac{d}{\lambda}\|\A(\xs-\xb)\|_1^2 \leq
\frac{d}{\lambda}\|\xs-\xb\|_1^2. 
\end{eqnarray}
The first inequality is $\log t\leq t-1$, and the second inequality is by the RIP-1 property of $\Phii$ (Eq.~(\ref{ripeq})) and definition of $\Gamma$ (Eq.~(\ref{gamma})).
\end{proof}
\end{lemma}
\begin{lemma}
\label{l4}
Let $\A$ be the expander sensing matrix, $\xs$ be the original almost $k$-sparse signal, and $\xh$ be a minimizer in Eq.~(\ref{min}). Then 
\begin{equation}
\label{measure}
\Ex\left[\|\A(\xs-\xh)\|_1\right] \le
 \sqrt{6}\min_{\xt \in \Gamma} \sqrt{\frac{d}{\lambda}}\|\xs-\xt\|_1+\sqrt{2\text{pen}(\xt)}
\end{equation}
\end{lemma}
\begin{proof}
 Lemmas \ref{l1}, \ref{l2}, and \ref{l3} together imply
\begin{eqnarray*}
\lefteqn{\Ex\left[\|\A(\xs-\xh)\|_1^2\right]}\\ \nonumber &\leq&
2\left(2+\frac{m\lambda}{d}\right)\min_{\xt\in\Gamma}\left(\frac{d}{\lambda}\|\xs-\xt\|_1^2+2\text{pen}(\xt)\right).
\end{eqnarray*}
Since $\lambda\ll \frac{1}{k\log(n/k)}$, and $m=O\left(k\log
  (n/k)\right)$, the ratio $\frac{m\lambda}{d}$ is much less than
$1$. So $2+\frac{m\lambda}{d}\ll 3$, and 
$$
\Ex\left[\|\A(\xs-\xh)\|_1^2\right] \leq
 6\min_{\xt\in\Gamma}\left(\frac{d}{\lambda}\|\xs-\xt\|_1^2+2\text{pen}(\xt)\right).
 $$
Now since the function $f(\xb)=\|A\xb+b\|_1^2$ is convex and the square
root function is strictly increasing, by applying Jensen's inequality
we get 
\begin{eqnarray*}
{\Ex\left[\|\A(\xs-\xh)\|_1\right]}\leq \sqrt{6}\min_{\xt \in \Gamma} \left(\sqrt{\frac{d}{\lambda}}\|\xs-\xt\|_1+\sqrt{2\text{pen}(\xt)}\right).
\end{eqnarray*}
\end{proof}
\begin{theorem}
Let $\A$ be the expander sensing matrix, $\lambda\ll
\frac{1}{k\log(n/k)}$ be a small positive value, $\xs$ be the original
almost $k$-sparse signal compressively sampled in the presence of
Poisson noise, $\xh$ be a minimizer in
Eq.~(\ref{min}), and $\that$ be the corresponding estimate in $\Theta$, i.e
$\xh=\that+\lambda\io$. Then 
\begin{eqnarray}\nonumber
\lefteqn{\Ex\left[\|\xs-\that\|_1\right]\leq\lambda m+
4\|\xs_{\overline{S}}\|_1+2\lambda m+3\sqrt{6}}\\ \nonumber 
&&\times \left(\min_{\ttil \in \Theta}
  \sqrt{\frac{d}{\lambda}}\left(\|\xs-\ttil\|_1+\lambda
    m\right)+\sqrt{2\hat{\text{pen}}(\ttil)}\right). 
\end{eqnarray}
\end{theorem}
\begin{proof}
In Lemma~\ref{l4}, we have bounded $\|\A(\xs-\xh)\|_1$. Now we can use
Theorem~\ref{piotr} to bound $\|\xs-\xh\|_1$. We have used a
$(2k,1/16)$-expander. Also since $\|\xs\|_1=1$, and any $x$ in
$\Gamma$ has the form $\theta+\lambda\io$ where $\|\theta\|_1=1$, and
$\|\io\|_1=m$, and since $\xh\in \Gamma$, we get 
$\|\xh\|_1\leq \|\that\|_1+\lambda\|\io\|_1$ and hence
$$\|\xs\|_1 \geq \|\xh\|_1-\lambda m.$$ As a result, by
Theorem~\ref{piotr} and Lemma~\ref{l4} we get 
\begin{eqnarray*}
\label{hoss}
\lefteqn{\Ex\left[\|\xs-\xh\|_1\right]\leq
  4\|\xs_{\overline{S}}\|_1+2\lambda m+}\\ 
\nonumber &&3\left(\sqrt{6}\min_{\xt \in \Gamma} \sqrt{\frac{d}{\lambda}}\|\xs-\xt\|_1+\sqrt{2\text{pen}(\xt)}\right).
\end{eqnarray*}
Consequently, we have derived a bound on how much $\xh$ differs from $\xs$. Since any $\xb$ in $\Gamma$ has the form $\tb+\lambda\io$ for some estimate $\tb$ in $\Theta$, using the triangle inequality we get
$$
\|\xs-\tb\|_1\leq \|\xs-\xb\|_1+\lambda\|\io\|_1=\|\xs-\xb\|_1+\lambda m,
$$
and so 
\begin{eqnarray*}
\nonumber \lefteqn{\Ex\left[\|\xs-\that\|_1\right]\leq\lambda m+ 4\|\xs_{\overline{S}}\|_1+2\lambda m+}\\
\nonumber&& 3\left(\sqrt{6}\min_{\ttil \in \Theta} \sqrt{\frac{d}{\lambda}}\left(\|\xs-\ttil\|_1+\lambda m\right)+\sqrt{2\hat{\text{pen}}(\ttil)}\right).
\end{eqnarray*}
\end{proof}
By substituting the values $m=O\left(k\log(n/k)\right)$, and $d=O\left(\log(n/k)\right)$, and choosing $$\lambda\ll \frac{1}{k\log(n/k)}$$ we can guarantee that $\Ex\left[\|\xs-\that\|_1\right]$ is of order 
\begin{eqnarray}
\label{becca}
 \|\xs_{\overline{S}}\|_1+\min_{\ttil \in \Theta} \left(\sqrt{k}\log\left(\frac{n}{k}\right)\|\xs-\ttil\|_1+\sqrt{2\hat{\text{pen}}(\ttil)}\right).
\end{eqnarray}
\modified{\begin{remark}
It has been shown by Willett \textit{et.al.} \cite{wr, rhmw} that, using random dense matrices, the MAP reconstruction algorithm can reconstruct a signal $\xs$ satisfying $\| \xs \|_1 = 1$ with the expected error of 
\begin{equation}\label{previous}\Ex\left[\|\xs-\that\|^2_2\right]\leq m\left[ \min_{\ttil \in \Theta} \|\xs-\ttil\|^2_2+\hat{\text{pen}}(\ttil)\right]+\frac{\log \frac{n}{m}}{m}.\end{equation}
Hence, for random dense matrices there is an $O\left(m^{-1}\right)$ min-max approximation error. This error cannot be made arbitrarily small by increasing the number of measurements as the first term in (\ref{previous}) also depends on $m$.  However, as stated earlier, the bounds of \cite{wr, rhmw} are not restricted to signals that are sparse in the canonical basis.\end{remark}
}

 \section{Experimental Results}
 \label{sec:experiment}
 To validate our results via simulation, we generated random sparse signals, simulated Poisson observations of the signal multiplied by the proposed expander graph sensing matrix, and reconstructed the signal using the proposed objective function in (\ref{min}). 

Each signal was a length $n=100,000$ signal with $k$ non-zero elements, where $k$ ranged from $1$ to $4,000$. Each of the non-zero elements was assigned intensity $I$, where $I$ was $10$, $100$, $1,000$, or $10,000$. The locations of the non-zero elements were selected uniformly at random for each trial.
The sensing matrix was a scaled adjacency matrix of an expander graph, as described earlier, with $d = 16$ and the number of rows $m=40,000$. 

Reconstruction was performed using a method described in
\cite{fesslerRecon} for reconstruction of sparse signals from indirect
Poisson measurements, precisely the situation encountered here. The
penalty function used in this implementation is proportional to
$\|\xb\|_1$; constructing a penalty function of this form which
satisfies the Kraft inequality is a subject of ongoing work. (The
authors would like to thank Mr. Zachary Harmany for his assistance
with the implementation of this algorithm.) After each trial, the
normalized $\ell_1$ error was computed as $\|\xs-\xh\|_1 / \|\xs\|_1$,
and the errors were averaged over $50$ trials.  The results of this
experiment are presented in Figure~\ref{fig:results}.

\begin{figure}
\begin{center}
\includegraphics[width=\columnwidth]{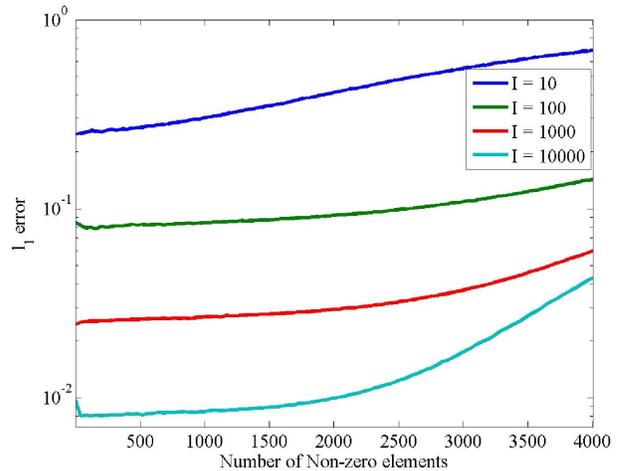}
\end{center}
\caption{Performance of reconstruction from Poisson measurements of
  expander CS data for different intensity and sparsity levels.}
\label{fig:results}
\end{figure}

 \section{Conclusions}
 \label{sec:conclusion}
 \modified{
 In this paper we investigated the advantages of expander-based sensing over  dense random sensing in the presence of Poisson noise.
 Even though Poisson model is essential in some applications, dealing with this noise model is challenging as the noise is not bounded, or even as concentrated as Gaussian noise, and is signal-dependent. 
 }\modified{
 Here we proposed using normalized adjacency matrices of expander graphs as an alternative construction of sensing matrices, and we showed that the binary nature and the RIP-1 property of these matrices yield provable consistency for a MAP reconstruction algorithm.
 }\modified{
 }
\section*{Acknowledgements}
\modified{The authors would like to thank Zachary Harmany for his assistance with the implementation of the reconstruction algorithm, and Piotr Indyk for his insightful comments on the performance of the expander graphs.}
 \bibliographystyle{unsrt} 
 \bibliography{main} 
\end{document}